\documentclass[11pt]{article}
\usepackage{amssymb}
\usepackage{amsfonts,amsmath, longtable}

\usepackage{comment}

\newtheorem{theorem}{Theorem}

\newcommand{\tr}{{\rm tr}}
\newcommand{\ti}[1]{\tilde{#1}}

\newcommand{\om}{\omega}

\newcommand{\Mat}{ {\rm Mat}_N }

\newcommand{\mC}{\mathbb C}
\newcommand{\mZ}{\mathbb Z}

\newcommand{\h}{\hbar}

\newtheorem{predl}{Proposition}[section]

\newenvironment{proof}{\par\noindent{\bf Proof.}}{\hfill$\scriptstyle\blacksquare$}

\def\beq{\begin{equation}}
\def\eq{\end{equation}}
\def\p{\partial}

\newcommand{\mats}[4]{\left(\begin{array}{cc}{#1}&{#2}\\ {#3}&{#4}
\end{array}\right)}

\def\res{\mathop{\hbox{Res}}\limits}

\begin{document}

\setcounter{page}{1}

\begin{center}


{\Large{\bf Large $N$ limit of spectral duality in classical}}

\vspace{3mm}

{\Large{\bf
integrable systems
}}



 \vspace{10mm}

 {\Large {R. Potapov}}
\qquad\quad\quad
 {\Large {A. Zotov}}

  \vspace{8mm}

 {\em Steklov Mathematical Institute of Russian
Academy of Sciences,\\ Gubkina str. 8, 119991, Moscow, Russia}

{\em Institute for Theoretical and Mathematical Physics,\\
Lomonosov Moscow State University, Moscow, 119991, Russia}
%


 {\small\rm {e-mails: trikotash@ya.ru, zotov@mi-ras.ru}}

\end{center}

\vspace{0mm}

\begin{abstract}
We describe the large $N$ limit of spectral duality between rational Gaudin models introduced
by Adams, Harnad and Hurtubise. The limit of the ${\rm gl}_N$ model is performed by means of
a noncommutative torus algebra represented by the fields on a torus with the Moyal-Weyl star product.
We apply the approach developed by Hoppe, Olshanetsky and Theisen to the Gaudin-type models
and describe the corresponding integrable field theory (2d hydrodynamics) on a torus.
The dual model is the large $N$ limit of
the ${\rm gl}_M$ Gaudin model with $N$ marked points written in the form of the Gaudin model
 with irregular singularities.
\end{abstract}

%

{\small{ \tableofcontents }}


\section{Introduction}\label{sec1}
\setcounter{equation}{0}
%
In this paper, we deal with classical integrable systems of Gaudin type \cite{G}.
The rational ${\rm gl}_N$ Gaudin model is described by an $N\times N$ Lax matrix
  \beq\label{q01}
  \begin{array}{c}
    \displaystyle{
 L(z)=V+\sum\limits_{a=1}^M\frac{S^a}{z-z_a}\in{\rm Mat}_N\,,\quad S^a\in{\rm Mat}_N\,,
 }
 \end{array}
 \eq
which is a rational ${\rm Mat}_N$-valued function of the spectral parameter $z$ -- a local coordinate
 on a punctured Riemann sphere ${\mathbb CP}^1\backslash\{z_1,...,z_M\}$.
It has simple poles at the marked points $z_1,...,z_M$ with residues $\res\limits_{z=z_a}L(z)=S^a$, $a=1,...,M$,
and $V={\rm diag}(v_1,...,v_N)$ is a constant diagonal twist matrix. The matrix elements of $S^a$ are dynamical variables.
The model is Hamiltonian. Its Poisson structure is a direct sum of the Poisson-Lie structures
on the Lie coalgebras ${\rm gl}_N^*$ corresponding to all residues $S^a$. Equations of motion are represented in the Lax form
\beq \label{q02}
\displaystyle{
{\dot L}(z)+[L(z),M(z)]=0\,,
}
\eq
where ${\dot L}(z)$ is a derivative with respect to time variable $t$ corresponding to some Hamiltonian.
The Hamiltonians can be computed from either $\tr L^k(z)$ or coefficients of the spectral curve
\beq\label{w03}
\displaystyle{
 \det\limits_{N\times N}(v-L(z))=0\,,
}
\eq
which is a complex curve in $\mC^2$ parameterized by $v$ and $z$.

Suppose there is another integrable model described by an $M\times M$ Lax matrix ${\tilde L}(v)$, and its
spectral curve
\beq\label{w04}
\displaystyle{
\det\limits_{M\times M}(z-\tilde{L}(v))=0
}
\eq
coincide with (\ref{w03}). Then these two models are called spectrally dual
to each other\footnote{In fact, besides the coincidence of the spectral curves one should also require
equality of Seiberg-Witten differentials and/or existence of the Poisson map between
 two phase spaces, see details
in \cite{MMZZ,MMZZR1}.}.
The first time this relation was observed for $N$-site closed Toda chain that has $2\times 2$ and $N\times N$ Lax representations \cite{FT}. For rational Gaudin models it was described by  Adams, Harnad and Hurtubise (see Theorem 2.1 in \cite{AHH}). Consider the
${\rm gl}_M$ Gaudin model of the form
  \beq\label{q03}
  \begin{array}{c}
    \displaystyle{
 {\ti L}(v)=Z+\sum\limits_{k=1}^N\frac{{\ti S}^k}{v-v_k}\in{\rm Mat}_M\,,\quad {\ti S}^a\in{\rm Mat}_M\,,
 }
 \end{array}
 \eq
where $Z={\rm diag}(z_1,...,z_M)$.
 The duality implies that the residues of  the Lax matrices
are rank one matrices of the following form:
\beq\label{w05}
    \displaystyle{
S^a_{ij}=\xi^{a}_{i}\eta^{a}_{j}\,,\quad i,j=1,...,N;\quad a=1,...,M\,,
}
\eq
and
\beq\label{w06}
\displaystyle{
{\ti S}^k_{ab}=\xi_{k}^{a}\eta_{k}^{b}\,,\quad k=1,...,N;\quad a,b=1,...,M\,.
}
\eq
Later this duality was generalized to a similar relation between the classical XXX spin chain and the trigonometric Gaudin model \cite{MMZZ,MMZZR1}, and also for a pair of the classical XXZ spin chains \cite{MMZZR2}.
Spectral duality has found numerous applications in studies of gauge theories and
the corresponding Seiberg-Witten geometry \cite{MM,MMZZ,N1}. In particular,
in \cite{MMZZ} the spectral duality was treated as a special limit of the AGT correpondence between conformal and gauge theories.
 The phenomenon of spectral duality has
different generalizations and interpretations in monodromy preserving equations, Knizhnik-Zamolodchikov
equations and related representation theory, see \cite{GVZ} and references therein.
It is also closely related to the Ruijsenaars action–angle duality, which can be described
through spectral duality \cite{PZ}.
At quantum level the spectral (or bispectral) dualities were described earlier by Mukhin, Tarasov and Varchenko
\cite{MTV}. An alternative description in terms multi-matrix models was suggested
by Bertola, Eynard, Harnad and Luu \cite{BEHL}.

\paragraph{Purpose of the paper.} Our main goal is to formulate and prove a large $N$ limit of the spectral duality
between Gaudin models.
The ${\rm gl}_N$ Gaudin model in this limit is described by infinite-dimensional Lax matrices.
Instead of working with an infinite-dimensional matrix space, we use the approach \cite{BHS,FFZ,HPS,JS,R}
 based on treating
$N \to \infty$ limit of $\Mat$
space in terms of the infinite-dimensional algebra
of noncommutative torus (NCT) $A_{\hbar}$. From the viewpoint of integrable systems,
the Lax matrices are then considered as $A_{\hbar}$ valued operators.
Then our aim is to formulate an infinite-dimensional analogue of rank 1 matrices (\ref{w05})
and give a definition of determinant
in order to define an infinite-dimensional analogue for the spectral curve (\ref{w03}).
For this purpose, we use two representations of the NCT algebra \cite{HOT}.
The first one is the representation in $\mathrm{Mat}_{\infty}(\mathbb{C}[[\hbar^{-1}, \hbar]])$
-- associative algebra of infinite-dimensional matrices with finitely many nonzero
diagonals\footnote{In fact, we deal with a smaller space -- a subspace in ${\rm Mat}_\infty$ with certain
convergence conditions, see (\ref{n8}).}.
The second one is the representation in the
algebra of smooth functions on a two-dimensional torus
\beq\label{w02}
\displaystyle{
A_{\hbar}\to\Big(C^{\infty}(\mathbb{T}^{2})[[\hbar,\hbar^{-1}]], \star\Big)
}
\eq
endowed with
the associative Moyal (or Weyl) deformation
quantization star-product $\star$ on $\mathbb{T}^{2}$
\cite{Weyl}:
\beq\label{w81}
  \begin{array}{c}
  \displaystyle{
f\star{g}=fg+\ti\hbar\sum\limits_{i,j=1}^2\vartheta^{ij}\partial_i{f}\partial_j{g}
+\frac{\ti\hbar^{2}}{2!}\sum\limits_{i,j=1}^2
\sum\limits_{k,l=1}^2\vartheta^{ij}\vartheta^{kl}\partial_i\partial_k{f}
\partial_j\partial_l{g}+\ldots=
}
\\ \ \\
  \displaystyle{
=\sum\limits_{n=0}\limits^{\infty}\frac{\ti\hbar^n}{n!}\sum\limits_{i_1,...,i_n;
j_1,...,j_n=1}^2\prod\limits^{n}\limits_{k=1}\vartheta^{i_k{j_k}}(\prod\limits_{k=1}
\limits^{n}\partial_{i_k})(f)
\times(\prod\limits_{k=1}\limits^{n}\partial_{j_k})(g)\,,
}
\end{array}
\eq
where $\vartheta^{ij}$ is the constant and skew-symmetric canonical Poisson structure on two-dimensional space,
and $\ti\hbar$ differs from $\hbar$ by some factor, which we fix below.
This representation provides the field-theoretical interpretation, since
the usage of (\ref{w02})-(\ref{w81}) leads to integrable field theories (or hydrodynamics) on $\mathbb{T}^{2} \times \mathbb{R}$.
The dynamical fields take values in $C^{\infty}(\mathbb{T}^{2})[[\hbar,\hbar^{-1}]]$, and this is the approach used for the description of
noncommutative field theories  \cite{DN}.
The same approach of the large $N$ limit can be applied to integrable systems \cite{HOT,Olsh}. In this way
certain integrable hydrodynamics with the Euler-type equations of motion arises \cite{Arnold}.
We apply a similar construction to describe the large $N$ limit in the ${\rm gl}_N$ Gaudin model.

Our first goal is to describe the infinite-dimensional limit of \eqref{q01} and obtain the corresponding integrable field theory on $\mathbb{T}^{2} \times \mathbb{R}$.
Both representations of $A_{\hbar}$ mentioned above  are related through the change of basis
in the space $\mathrm{Mat}_{\infty}(\mathbb{C}[[\hbar^{-1}, \hbar]])$.
The first one is the standard basis and the second is the
one used in the description of sin-algebra. The standard basis is needed since the rank one matrices (\ref{w05})
are easily written in this basis, and the second one is naturally related to the representation (\ref{w02}).
Put it differently, we use the following embeddings:
$
	(C^{\infty}(\mathbb{T}^{2})[[\hbar^{-1},\hbar]], \star)
 \hookrightarrow A_\hbar \hookrightarrow \mathrm{Mat}_{\infty}(\mathbb{C}[[\hbar^{-1}, \hbar]])
$
and treat them as faithful representations.
We also use an infinite-dimensional analogue of the classical $r$-matrix structure to ensure the integrability
of the described large $N$ limit.


Finally, we define the spectral curve for the infinite-dimensional Gaudin model (the large $N$ limit in the ${\rm gl}_N$ Gaudin model). It is written in the form of power series in spectral parameter. The dual model is described by the
Laurent power series of finite-dimensional matrix variables.
It corresponds to a certain limit in the Gaudin model with irregular singularities.
This type of models (where the Lax matrix has higher-order poles) was introduced by Feigin, Frenkel, Reshetikhin and Laredo \cite{FFRL}.
The spectral dualities in Gaudin models
with irregular singularities were studied in \cite{VY} at the quantum level. In our paper, we do not
consider the gluing of marked points (or twist parameters) and the corresponding degenerations described by Jordan blocks. At the same time, the dual model in our description is given by the Lax matrix with higher-order poles.

In summary, our results establish a relation between a certain class of 2+1 integrable field theories 
(including those of hydrodynamical type) and the large $N$ limit in Gaudin models with irregular singularities.

The paper is organized as follows. In the first Section we define the NCT algebra and review its algebraic properties. In the
next Section we describe the infinite-dimensional limit for the Gaudin models. In Section \ref{sec4} the large $N$ limit of the spectral duality in Gaudin models is proposed.

\section{Algebra of noncommutative torus}\label{sec2}
\setcounter{equation}{0}
In this section we recall basic definitions and properties related to the NCT algebra. Here we mostly follow \cite{HOT,Olsh}.
It is explained below that the NCT algebra may be represented using infinite-dimensional matrices
and how it arises in the deformation quantization on a torus.

The NCT algebra $A_{\h}$ is defined
 as an associative algebra over the ring of formal Laurent series $\mathbb{C}[[\hbar^{-1}, \hbar]]$:
 \footnote{By $\langle\langle U_{1},U_{2}\rangle\rangle$ we mean a free algebra of infinite series.}
	\beq\label{n1}
	\displaystyle{A_{\hbar} = \mathbb{C}[[\hbar^{-1}, \hbar]]\langle\langle U_{1},U_{2}\rangle\rangle/\langle U_{1}U_{2} - \omega U_{2}U_{1}\rangle, \,\,\,\,\, \omega = e^{4\pi i \hbar}\,.
	}
	\eq
It is convenient to introduce a special basis in $A_{\h}$ \cite{FFZ}:
\beq\label{n2}
\displaystyle{ T_{\vec{m}} = \frac{i}{4\pi\hbar} \omega^{\frac{m_{1}m_{2}}{2}} U_{1}^{m_{1}}U_{2}^{m_{2}}\,,
\quad {\vec m}=(m_1,m_2)\in\mZ\times\mZ\,.
}
\eq
Then the multiplication law takes the form
\beq\label{n3}
\displaystyle{
T_{\vec{m}}T_{\vec{n}}  = \frac{i}{4 \pi \hbar}
	\omega^{-\frac{1}{2}\,\vec{m}\times\vec{n}}T_{\vec{m}+\vec{n}}\,,
}
\eq
where $\vec{m}\times\vec{n}=m_1n_2-m_2n_1$.

\subsection{Finite-dimensional representation and related structures}
Assume for now that $\h \in \mathbb{R}$, and $A_{\h}$ is an algebra over $\mathbb{C}$.
The defining relation is
\beq\label{n31}
\displaystyle{
U_{1}U_{2} = \omega U_{2}U_{1}\,.
}
\eq
 When $\hbar$ is rational ( $\hbar = \frac{m}{n} \in \mathbb{Q}$ )  the algebra (\ref{n1}) has a finite-dimensional representation $\rho$ in $\mathrm{Mat}_{n} (\mathbb{C})$ \cite{Belavin}. Let for simplicity $\omega=\exp(2\pi i/n)$.
Then
\beq\label{n6}
\rho(U_{1})=Q=
\left(\begin{array}{ccccc}
1&0&0&\cdots&0\\
0&\omega&0&\cdots&0\\
0&0&\omega^2&\cdots&0\\
\vdots&\vdots&\ddots&\ddots&\vdots\\
0&0&0&\cdots&\omega^{n-1}
\end{array}\right)\,,\quad
\rho(U_{2}) = \Lambda=
\left(\begin{array}{ccccc}
0&1&0&\cdots&0\\
0&0&1&\cdots&0\\
\vdots&\vdots&\ddots&\ddots&\vdots\\
0&0&0&\cdots&1\\
1&0&0&\cdots&0
\end{array}\right)
\eq
and $Q\Lambda = \om \Lambda Q$.
The basis (\ref{n2}) then corresponds to the following basis in ${\rm Mat}_n$:
\beq\label{cg5}
\displaystyle{
	\rho(T_{\vec{a}}) = I_{\vec{a}} =e^{i \pi \frac{a_{1}a_{2}}{n}}Q^{a_{1}}\Lambda^{a_{2}}\,,
\quad a_1,a_2\in\mZ/n\mZ\,.
}
\eq
In particular, $I_{(0,0)}=1_n$ is the identity $n\times n$ matrix.
Originally, this basis was used in \cite{Belavin} to construct the classical
elliptic $r$-matrix\footnote{In this paper we do not deal with elliptic models. The large $N$ limit
of elliptic models will be studied in our next paper. Here we just mention
that the finite-dimensional representation (\ref{n6}) of the Heisenberg group (\ref{n31})
is known to be naturally related to elliptic functions.}
\beq\label{cg2}
\displaystyle{
	r_{12}(z) = \frac{\vartheta'(z)}{n\vartheta(z)} 1_n\otimes 1_n+
 \sum\limits_{\vec{a} \in \mZ^{2}_{n}/\{0,0\}}
I_{\vec{a}}\otimes I_{-\vec{a}}\,
\exp\Big(2\pi i\frac{a_{2}}{n}  z\Big)\frac{\vartheta'(0)\vartheta(z + \frac{a_{1}+a_{2}\tau}{n})}{n\vartheta(z)\vartheta( \frac{a_{1}+a_{2}\tau}{n})}\,,
}
\eq
where $\vartheta(z)$ is the first Jacobi theta function on elliptic curve with moduli $\tau$.
It is a solution of the
classical Yang-Baxter equation
\beq\label{cg1}
\displaystyle{
	[r_{12}(z_{1}-z_{2}),r_{13}(z_{1}-z_{3})]+[r_{12}(z_{1}-z_{2}),r_{23}(z_{2}-z_{3})] + [r_{13}(z_{1}-z_{3}),r_{23}(z_{2}-z_{3})] = 0\,.
}
\eq
In the rational limit (${\rm Im}(\tau)\to 0$) the Belavin-Drinfeld $r$-matrix (\ref{cg2}) turns into
the rational XXX $r$-matrix
\beq\label{cg21}
\displaystyle{
	r_{12}(z) =\frac{P_{12}}{z}\,,
}
\eq
where $P_{12}$ is the matrix permutation operator
\beq\label{n221}
\displaystyle{
P_{12}=\sum\limits_{i,j=1}^n E_{ij}\otimes E_{ji}=\frac{1}{n}\sum\limits_{a\in\mZ_n\times\mZ_n}I_a\otimes I_{-a}\,.
}
\eq
In order to define the Hamiltonian mechanics for the Gaudin type models the Poisson-Lie structure
on the Lie coalgebra ${\rm gl}_n^*$ is required. For a dynamical matrix variable $S=\sum_{ij} E_{ij}S_{ij}$ the brackets
\beq\label{cg22}
\displaystyle{
\{S_{ij},S_{kl}\}=S_{kj}\delta_{il}-S_{il}\delta_{kj}\,,\quad i,j,k,l=1,...,n
}
\eq
(constructed through the Lie algebra relations $[E_{ij},E_{kl}]=E_{il}\delta_{kj}-E_{kj}\delta_{il}$) are
written in the tensor form (see \cite{FT})
\beq\label{cg23}
\displaystyle{
	\{S_1,S_2\}=[P_{12},S_1]\,,\qquad S_1=S\otimes 1_n\,,\ S_2=1_n\otimes S\,.
}
\eq
Using the finite-dimensional analogue of (\ref{n3})
$I_{\vec m}I_{\vec k}=\exp\Big(\frac{\pi i}{n}(k_1m_2-k_2m_1)\Big)I_{{\vec m}+{\vec k}}$ one  easily gets
\beq\label{cg24}
\displaystyle{
[I_{\vec m},I_{\vec k}]=2i\sin\Big(\frac{\pi i}{n}(k_1m_2-k_2m_1)\Big)I_{{\vec m}+{\vec k}}\,,
\qquad m_1,m_2,k_1,k_2\in\mZ/n\mZ\,.
}
\eq
Then (\ref{cg23}) provides the Poisson brackets
\beq\label{cg25}
\displaystyle{
\{S_{\vec m},S_{\vec k}\}=\frac{2i}{n}\sin\Big(\frac{\pi i}{n}(k_1m_2-k_2m_1)\Big)S_{{\vec m}+{\vec k}}\,.
}
\eq
Notice that one may obtain the Lie algebra of $A_{\h}$ as a certain infinite-dimensional limit of $\mathfrak{gl}_{n}(\mathbb{C})$ \cite{BHS}.
The commutation relations (\ref{cg24}) are the finite-dimensional analogue of the sin-algebra, which we consider below. The brackets (\ref{cg25}) will be generalized to the infinite-dimensional case.

 \subsection{Infinite-dimensional matrix representation}
 Notice that the structure constants \eqref{n3} depend on $\h$. This dependence makes the algebras $A_{\h}$ with different values of $\h$ be nonisomorphic to each other \cite{JS}.
When $\hbar\notin \mathbb{Q}$ the algebra \eqref{n1} can be embedded into the algebra of infinite-dimensional matrices with finitely many nonzero diagonals $\mathrm{Mat}_{\infty}(\mathbb{C}[[\hbar^{-1}, \hbar]])$,
that is $A_{\h}$ has representation in $\mathrm{Mat}_{\infty}(\mathbb{C}[[\hbar^{-1}, \hbar]])$.

 Choose the basis of matrix units $E_{ab}$  in  $\mathrm{Mat}_{\infty}(\mathbb{C}[[\hbar^{-1}, \hbar]])$.
 Then the above mentioned faithful representation
 is given by the map:
	\beq\label{n7}
	\displaystyle{
		W:\,\,\,T_{m_{1}m_{2}} \to \frac{i}{4\pi \hbar}\om^{\frac{m_{1}m_{2}}{2}}\sum\limits_{k \in \mathbb{Z}}\om^{m_{1}k}E_{k,k+m_{2}}\,.
	}
	\eq
In order to invert this map one should take smaller a subalgebra of $\mathrm{Mat}_{\infty}(\mathbb{C}[[\hbar^{-1}, \hbar]])$. In particular, we need matrices with convergent traces. Consider
 the subalgebra $\mathrm{Mat_{conv}}(\mathbb{C})$, which consists of
matrices $\sum\limits_{a,b \in \mathbb{Z}}M_{ab}E_{ab}$, $M_{ab} \in \mathbb{C}$ with the property
\beq\label{n8}
\displaystyle{\sup\limits_{ab}(1+a^{2}+b^{2})^{k}|M_{ab}|^{2} < \infty\, \,\,\,\,\, \forall k \in \mathbb{N}\,.
}
\eq
Then the inverse of the map (\ref{n7}) is as follows:
	\beq\label{n9}
	\displaystyle{
	W^{-1}:\,\,\,	E_{ab} \to \frac{4 \pi \h}{i}\sum\limits_{k \in \mathbb{Z}}\om^{-\frac{a+b}{2}k}T_{k,b-a}\,.
	}
	\eq
	%
%
%
Notice that $T_{k,b-a}$ in the r.h.s. of (\ref{n9}) are rather infinite-dimensional matrices
$W(T_{k,b-a})$ from the space $\mathrm{Mat_{conv}}(\mathbb{C})$ than elements of $A_\hbar$.
 In what follows we do not use different notations for different representations of $T_{\vec m}$,
 they are stripped down but clear in context.

To summarize, we deal with the infinite-dimensional matrix space $\mathrm{Mat_{conv}}(\mathbb{C})$ (\ref{n8})
and use two bases: the standard one $E_{ab}$ and $T_{\vec m}$ (\ref{n7}).
The relationship between the bases is as follows:
\beq\label{dn1}
\displaystyle{
	T_{\vec{m}} = \sum\limits_{i,j \in \mZ}W_{\vec{m},ij}E_{ij}\,,\,\,\,\,\,\,\,\,\,\, E_{ab} = \sum\limits_{
    \vec{n}\in \mZ^{2}}W^{-1}_{ab,\vec{n}}T_{\vec{n}}\,,
}
\eq
where
\beq\label{n10}
\displaystyle{W_{\vec{m},ij} = \frac{i}{4\pi\h}w^{\frac{m_{1}m_{2}}{2}}w^{m_{1}i}\delta_{j,i+m_{2}}	
}
\eq
and
\beq\label{n11}
\displaystyle{
	W_{ij,\vec{n}}^{-1} = \frac{4 \pi \h}{i}w^{-\frac{i+j}{2}n_{1}}\delta_{n_{2},j-i}\,.
}
\eq
In order to construct the Hamiltonians in integrable systems we also need to define the trace:
%
\beq\label{n12}
\displaystyle{
	\tr T_{m_{1}m_{2}} = \frac{i}{4 \pi \h}\delta_{m_{1},0}\delta_{m_{2},0}
}
\eq
with the property
\beq\label{n13}
\displaystyle{
	\tr (AB) = \tr (BA)\,,\quad A,B\in\mathrm{Mat_{conv}}(\mathbb{C})\,.
}
\eq
Being restricted to $\mathrm{Mat_{conv}}(\mathbb{C})$, this definition provides the usual matrix trace.

 \subsection{Infinite-dimensional representation via deformation quantization on a torus}
The field-theoretical interpretation of integrable models on $A_{\h}$ is related to
another representation (\ref{w02}) of $A_\hbar$. Consider the space
of smooth functions on a torus $C^{\infty}(\mathbb{T}^{2})$. It is endowed with the canonical Poisson structure
\beq\label{n14}
\displaystyle{ \left\{f,g\right\}(\phi_{1},\phi_{2}) = \p_{\phi_{1}}f\p_{\phi_{2}}g - \p_{\phi_{2}}f\p_{\phi_{1}}g\,, \,\,\,\,\,\,\,\,\,\, f,g \in C^{\infty}(\mathbb{T}^{2})\,.
}
\eq
The deformation quantization of \eqref{n14} via the Moyal-Weyl star product is given
by (\ref{w81}), where we fix $\ti\hbar=2\pi i \hbar$. It can be written in short as
\beq\label{n15}
\displaystyle{ f \star g = m \circ e^{\frac{1}{2}\,4\pi i \h \left(\p_{\phi_{1}}\otimes\p_{\phi_{2}} - \p_{\phi_{2}}\otimes\p_{\phi_{1}}\right)} f\otimes g\,,
}
\eq
where $m(a\otimes b) = ab$.
This product makes $(C^{\infty}(\mathbb{T}^{2})[[\h]], \star)$ an associative algebra with the property
\beq\label{n16}
\displaystyle{ \left[f,g\right]_{\star} = 4\pi i \h \left\{f,g\right\} + O(\h^2)\,.
}
\eq
Details of the deformation quantization on $\mathbb{T}^{2}$ can be found in \cite{R}.
Let $f \in C^{\infty}(\mathbb{T}^{2})$. Its Fourier series is of the form
\beq\label{n17}
\displaystyle{ f = \sum\limits_{\vec{m} \in \mathbb{Z}^{2}}f_{\vec{m}}e^{i \vec{m} \vec{\phi}}\,.	
}
\eq
The statement on the existence of the representation (\ref{w02}) is based on the map
\beq\label{n19}
\displaystyle{T_{\vec{m}} \to \frac{i}{4 \pi \h}\,e^{i \vec{m} \vec{\phi}}	
}
\eq
and the property of (\ref{n15})
\beq\label{n18}
\displaystyle{
	e^{i \vec{m} \vec{\phi}} \star 	e^{i \vec{n} \vec{\phi}} = \om^{-\frac{1}{2}\, \vec{m} \times \vec{n}}e^{i (\vec{m}+\vec{n})\vec{\phi}}\,.
}
\eq
%
%
%
For $f \in C^{\infty}(\mathbb{T}^{2})[[\h^{-1},\h]]$ the trace \eqref{n12} takes the form
\beq\label{n20}
\displaystyle{ \tr f = \frac{1}{4 \pi^{2}}\int \limits_{\mathbb{T}^{2}}\mathrm{d}\vec{\phi} f(\vec{\phi})\,.
}
\eq

\subsection{Sin-algebra and Poisson structure}

\paragraph{Sin-algebra and its coalgebra.}
The Lie algebra $L_{\h}$  of $A_{\h}$ is defined through commutation relations. Its structure constants in the basis \eqref{n2} are written as the sin-algebra \cite{FFZ}:
\beq\label{n21}
\displaystyle{ [T_{\vec{m}},T_{\vec{n}}] = \frac{1}{2 \pi \h} \sin (2 \pi \h \,\vec{m} \times \vec{n})T_{\vec{m}+\vec{n}}\,,
\qquad \vec{m} \times \vec{n} = m_{1}n_{2} - n_{1}m_{2}\,.
}
\eq
Define the Poisson-Lie structure on  the dual space $L^{*}_{\h}$ as
\beq\label{n22}
\displaystyle{
	\left\{S_{\vec{m}},S_{\vec{n}}\right\} = \frac{1}{2 \pi \h} \sin (2 \pi \h \, \vec{m} \times \vec{n})S_{\vec{m}+\vec{n}}\,,
}
\eq
%
Let us represent it in a form similar to (\ref{cg23}).
For this purpose consider the infinite-dimensional analogue of the permutation operator (\ref{n221})
	\beq\label{f9}
	\displaystyle{  \tilde{P}_{12} =
\sum\limits_{\vec{m} \in \mathbb{Z}^{2}} T_{\vec{m}} \otimes T_{-\vec{m}}\in L_\hbar^{\otimes 2}\,.
	}
	\eq
Notice that \eqref{f9} is obtained by $n \to \infty$ limit of $nP_{12}$.
%
\begin{predl}\label{nl3}
	\beq\label{n23}
	\displaystyle{
		\left\{S_{1},S_{2}\right\} = [\tilde{P}_{12},S_{1}]\,, \qquad S_1=S\otimes T_{(0,0)}\,,\ S_2= T_{(0,0)}\otimes S\,,
	}
	\eq
	where $\tilde{P}_{12}$ is given by (\ref{f9}).	
\end{predl}
\begin{proof}
	The proof is very similar to the finite-dimensional computation:
	\beq\label{n25}
	\begin{array}{c}
		\displaystyle{
			[\tilde{P}_{12},S_{1}] = \sum\limits_{\vec{m},\vec{a} \, \in \mathbb{Z}^{2}}S_{\vec{m}}\left[T_{\vec{a}},T_{\vec{m}}\right]\otimes T_{-\vec{a}} =
		}
		\\ \ \\
		\displaystyle{
			= \sum\limits_{\vec{m},\vec{a} \, \in \mathbb{Z}^{2}}S_{\vec{m}}\frac{1}{2\pi \h}\sin(2\pi \h \, \vec{a} \times \vec{m})T_{\vec{m} + \vec{a}}\otimes T_{-\vec{a}} = \sum\limits_{\vec{m},\vec{n} \, \in \mathbb{Z}^{2}}S_{\vec{m} + \vec{n}}\frac{1}{2\pi \h}\sin(2\pi \h \, \vec{m} \times \vec{n})T_{\vec{m}}\otimes T_{\vec{n}}\,.
		}
	\end{array}
	\eq
By definition, the l.h.s. of (\ref{n23}) equals
\beq\label{n24}
\displaystyle{
	\left\{S_{1},S_{2}\right\} = \sum\limits_{\vec{m},\vec{n} \, \in \mathbb{Z}^{2}}
\left\{S_{\vec{m}},S_{\vec{n}}\right\}
 T_{\vec{m}}\otimes T_{\vec{n}}
}
\eq
and this finishes the proof.
\end{proof}

As a result we have infinitely many functions commuting with respect to the Poisson brackets (\ref{n22}):
\beq\label{n26}
\displaystyle{
	\{\tr S^{k} , \tr S^{l}\} = 0\,, \,\,\,\, \forall k\,,l \in \mathbb{N}\,.
}
\eq
These are, in fact, the Casimir functions of \eqref{n22}.

\paragraph{Field theory on a torus.}
Using the representation \eqref{n19} we may rewrite the upper relations (\ref{n21})-(\ref{n26})
in terms of functions on a torus. Indeed, it follows from (\ref{n18}) that
\beq\label{n261}
\displaystyle{
[e^{i \vec{m} \vec{\phi}},e^{i \vec{n} \vec{\phi}}]_\star\equiv
e^{i \vec{m} \vec{\phi}} \star 	e^{i \vec{n} \vec{\phi}} - e^{i \vec{n} \vec{\phi}} \star e^{i \vec{m} \vec{\phi}} 	
= -2i\sin (2 \pi \h \,\vec{m} \times \vec{n})e^{(\vec{m}+\vec{n})\vec{\phi}}\,.	
}
\eq
Therefore, for
\beq\label{n262}
\displaystyle{
T_{\vec m}(\vec{\phi})=	\frac{i}{4 \pi \h}\,e^{i \vec{m} \vec{\phi}}	
}
\eq
we have
\beq\label{n263}
\displaystyle{
[T_{\vec m}(\vec{\phi}),T_{\vec n}(\vec{\phi})]_\star=
\frac{1}{2 \pi \h} \sin (2 \pi \h \,\vec{m} \times \vec{n})T_{\vec{m}+\vec{n}}(\vec{\phi})\,.
}
\eq
Next, consider a field on a torus
\beq\label{n264}
\displaystyle{
S(\vec{\phi})= \sum\limits_{\vec{m}\in \mZ^{2}}e^{i \vec{m}\vec{\phi}}S_{\vec m}\,.
}
\eq
The Poisson structure for these fields follows from (\ref{n22}).

\begin{predl}\label{nl4}
	\beq\label{n28}
	\displaystyle{
		\{S(\vec{\phi}), S(\vec{\theta})\} = \frac{i}{4 \pi \h}[S(\vec{\phi}),\delta(\vec{\theta}-\vec{\phi})]_{\star} = \frac{2i}{\pi^{2}(4\pi \h)^{4}} \int_{\mathbb{T}^{2}}\mathrm{d}\vec{\psi}S(\vec{\psi})\sin(\frac{1}{2 \pi \h}|\vec{\phi}\vec{\psi}\vec{\theta}|)\,,
	}
	\eq
	where
	\beq\label{n29}
	\displaystyle{
		|\vec{\phi}\vec{\phi'}\vec{\phi''}| = \vec{\phi} \times \vec{\phi'}+\vec{\phi'} \times \vec{\phi''}+\vec{\phi''} \times \vec{\phi}\,,
	}
	\eq
and
	\beq\label{n30}
	\displaystyle{
		\delta(\vec{\phi}) = \sum\limits_{\vec{m} \in \mZ^{2}}e^{i\vec{m}\vec{\phi}}\,.
	}
	\eq
\end{predl}
\begin{proof}
	Let us first comment on the star-commutator in \eqref{n28}. The star product \eqref{n15} is defined for a pair of functions at the same point. In \eqref{n28} we treat the variable $\vec{\theta}$ as a parameter, and calculate the star product at point $\vec{\phi}$, i.e.
	\beq\label{n311}
	\displaystyle{
		S(\vec{\phi}) \star \delta(\vec{\phi} - \vec{\theta}) = m \circ e^{\frac{1}{2}4\pi i \h \left(\p_{\phi_{1}}\otimes\p_{\phi_{2}} - \p_{\phi_{2}}\otimes\p_{\phi_{1}}\right)} S(\vec{\phi})\otimes \delta(\vec{\phi} - \vec{\theta})\,.
	}
	\eq
With this remark we have
	\beq\label{n33}
	\begin{array}{c}
		\displaystyle{
			\{S(\vec{\phi}), S(\vec{\theta})\}=  \sum\limits_{\vec{m},\vec{n} \in \mZ^{2}}e^{i \vec{m}\vec{\phi}}e^{i \vec{n}\vec{\theta}}\left\{S_{\vec{m}},S_{\vec{n}}\right\} = \sum\limits_{\vec{m},\vec{n} \in \mZ^{2}}e^{i \vec{m}\vec{\phi}}e^{i \vec{n}\vec{\theta}}\frac{1}{2 \pi \h} \sin (2 \pi \h \, \vec{m} \times \vec{n})S_{\vec{m}+\vec{n}} =
		}
		\\ \ \\
		\displaystyle{
			= \sum\limits_{\vec{m},\vec{n} \in \mZ^{2}}S_{\vec{m}}e^{i \vec{m}\vec{\phi}}e^{i \vec{n}(\vec{\theta}-\vec{\phi})}\frac{1}{2 \pi \h} \sin (2 \pi \h \, \vec{m} \times \vec{n}) = \sum\limits_{\vec{m},\vec{n} \in \mZ^{2}}S_{\vec{m}}e^{i \vec{n}\vec{\theta}}\frac{i}{4 \pi \h}[e^{i \vec{m}\vec{\phi}},e^{-i \vec{n}\vec{\phi}}]_{\star} =
		}
		\\ \ \\
		\displaystyle{
			= \frac{i}{4 \pi \h}[S^{a}(\vec{\phi}),\delta(\vec{\theta}-\vec{\phi})]_{\star}\,.
		}
	\end{array}
	\eq
	The integral form of the Poisson brackets in \eqref{n28} is obtained from the integral representation of the Moyal-Weyl product
	\beq\label{n32}
	\displaystyle{
		f \star g \,(\vec{\phi}) = - \frac{i}{\pi^{2}(4 \pi \h)^{3}} \int_{\mathbb{T}^{2} \times \mathbb{T}^{2}}\mathrm{d} \vec{\phi'} \mathrm{d} \vec{\phi''} e^{\frac{2i}{4 \pi \h}|\vec{\phi}\vec{\phi'}\vec{\phi''}|}f(\vec{\phi}')g(\vec{\phi''})\,.
	}
	\eq
	One should substitute the delta-function into \eqref{n32}.
\end{proof}

Notice that the Poisson bracket (\ref{n28}) is non-ultralocal due to the presence of derivatives of the delta-function.
Therefore, the Poisson brackets of fields at different points are nonzero.
In this representation the relation \eqref{n26} holds true for the functions
\beq\label{n34}
\displaystyle{
	\tr S^{k} = \frac{1}{4\pi^{2}}\int _{\mathbb{T}^{2}}\mathrm{d} \vec{\phi}\,S^{\star k }(\vec{\phi})\,,
\qquad S^{\star k }(\vec{\phi})=\underbrace{S(\vec{\phi})\star\dots\star S(\vec{\phi})}_{\hbox{$k$ times}}\,.
}
\eq

One may also consider the ''classical limit'' $\h \to 0$. In this limit the bracket \eqref{n28} turns into
\beq\label{n35}
\begin{array}{c}
	\displaystyle{
		\{S(\vec{\phi}), S(\vec{\theta})\}_{0} = \lim \limits_{\h \to 0}\{S(\vec{\phi}), S(\vec{\theta})\} = \partial_{\phi_{1}}S(\vec{\phi})\partial_{\phi_{2}}\delta(\vec{\theta}-\vec{\phi}) - \partial_{\phi_{2}}S(\vec{\phi})\partial_{\phi_{1}}\delta(\vec{\theta}-\vec{\phi})\,,
	}
\end{array}
\eq
and the Casimir functions take the form
\beq\label{n36}
\displaystyle{
	\tr S^{k} = \frac{1}{4\pi^{2}}\int _{\mathbb{T}^{2}}\mathrm{d} \vec{\phi}\,S^{k }(\vec{\phi})\,.
}
\eq

Below we describe the large $N$ limit of the Gaudin model. From viewpoint of the above description
we deal with a field theory on $\mathbb{T}^{2} \times \mathbb{R}$, where the real line is for time variable.
A dynamical field variable $S \in C^{\infty}(\mathbb{T}^{2} \times \mathbb{R})$ is then written as follows:
\beq\label{n27}
\displaystyle{
	S(\vec{\phi},t) = \sum\limits_{\vec{a} \in \mZ^{2}}S_{\vec{a}}(t)e^{i \vec{a}\vec{\phi}}\,.
}
\eq

\section{Large $N$ limit in ${\rm gl}_N$ Gaudin models}\label{sec3}
\setcounter{equation}{0}
In this Section we recall the description of the classical ${\rm gl}_N$ Gaudin model and study its
 large $N$ limit using the field theory approach (\ref{n27}) developed in \cite{HOT,Olsh}.

\subsection{Classical $r$-matrix}
 Let us formulate the Lax representation in terms of the classical $r$-matrices, which are solutions to the
 classical Yang-Baxter equation (\ref{cg1}).
 \paragraph{The finite-dimensional case.}
For any classical $r$-matrix\footnote{We deal with non-dynamical and skew-symmetric $r$-matrices, that
is $r_{12}(z)=-r_{21}(-z)$ holds.} solving (\ref{cg1}) one can associate the Gaudin model with $M$ marked points
 in the following way.
The phase space
	\beq\label{cg70}
	\displaystyle{
		\underbrace{\mathfrak{gl}_{N}^{*}\times ... \times \mathfrak{gl}_{N}^{*}}_M
	}
	\eq
 is
endowed with the Poisson-Lie brackets
	\beq\label{cg7}
	\displaystyle{
		\{S_{1}^{a},S_{2}^{b}\} = \delta^{ab} \left[P_{12},S^{a}_{1}\right]\,,
	}
	\eq
	and the Lax matrix has the form
	\beq\label{cg8}
	\displaystyle{
		L(z) = \sum\limits_{k = 1}^M\tr_{2}\Big(r_{12}(z-z_{k})S^{k}_{2}\Big)\,,
	}
	\eq
where $\tr_{2}$ is the trace over the second tensor component. With the rational XXX $r$-matrix (\ref{cg21})
the Lax matrix (\ref{cg8}) turns into (\ref{q01}) with $V=0$.
Due to the classical Yang-Baxter equation \eqref{cg1} the Lax matrix (\ref{cg8}) satisfies the linear $r$-matrix structure
\beq\label{cg9}
\displaystyle{
	\{L_{1}(z),L_{2}(w)\} = \left[ r_{12}(z-w),L_{2}(z)+L_{1}(w)\right]\,,
}
\eq
which guarantees that traces of powers of $L(z)$ are Poisson commuting Hamiltonians:
\beq\label{cg10}
\displaystyle{
	\{\tr L^{k}(z), \tr L^{m}(w)\} = 0\,, \,\,\,\, \forall\, k\,,m \in \mathbb{N}\,.
}
\eq
The constant twist matrix $V$ can be added to $L(z)$ if $\left[ r_{12}(z-w),V_{1}+V_{2}\right]=0$, and
this is the case for $r$-matrix (\ref{cg21}).
Below we describe $V=0$ case for the sake of brevity.
The Gaudin Hamiltonians are defined as
\beq\label{cg11}
\displaystyle{
	H_{i} = \frac{1}{2}\res \limits_{z = z_{i}}\tr L^{2}(z)\,,\quad i=1,...,M\,
}
\eq
For the classical $r$-matrices under consideration we have the property
\beq\label{cg111}
\displaystyle{
	\res\limits_{z=0}r_{12}(z)=P_{12}\,.
}
\eq
Then
\beq\label{cg112}
\displaystyle{
	H_{i} = \sum\limits_{k:k\neq i}^N\tr_{1,2}\Big(r_{12}(z_i-z_k)S_2^kS_1^i\Big)\,,\quad i=1,...,M\,.
}
\eq

 \paragraph{The infinite-dimensional case.} Notice that the above given definition uses the commutation relations in
 the Lie algebra (together with the Poisson brackets in the Lie coalgebra) and the existence of an invariant trace.
 These structures in the infinite-dimensional case are given by (\ref{n21}), (\ref{n22}), (\ref{n263}) and (\ref{n12}), (\ref{n20}) respectively. This allows to perform straightforward generalization to the large $N$ limit.
In particular the classical rational $r$-matrix \eqref{cg21} is replaced with the following one
\beq\label{rg1}
\displaystyle{
	r_{12}(z) = \frac{1}{z}\sum\limits_{\vec{a} \in \mZ^{2}}T_{\vec{a}}\otimes T_{-\vec{a}}
  \in L_{\h}^{\otimes 2} \,.
}
\eq
It satisfies the classical Yang-Baxter equation \eqref{cg1}, and the proof is a direct generalization
of the finite-dimensional one.

Let us define the
	$A_{\h}$ rational Gaudin model as a model on the phase space  $\underbrace{L^{*}_{\h}\times ... \times L^{*}_{\h}}_\text{$M$}$ with the Poisson-Lie structure
	\beq\label{rg2}
	\displaystyle{
		\{S_{\vec{m}}^{a},S_{\vec{n}}^{b}\} = \frac{\delta^{ab}}{2 \pi \h} \sin (2 \pi \h \, \vec{m} \times \vec{n})S_{\vec{m}+\vec{n}}^{a}\,.
	}
	\eq
	The Lax matrix
	\beq\label{rg3}
	\displaystyle{
		L(z) = \sum\limits_{a =1}^M \frac{S^{a}}{z-z_{a}} \in L_{\h}\,,\qquad S^{k} =\sum\limits_{\vec{a} \in \mZ^{2}}S^{k}_{\vec{a}}T_{\vec{a}}\in L_{\h}
	}
	\eq
can be considered in the (at least) two representations -- in  $\mathrm{Mat}_{\infty}(\mathbb{C}[[\hbar^{-1}, \hbar]])$ (\ref{n7}) with the trace (\ref{n12}), or in the field case (\ref{w02}) with the trace (\ref{n20}).
Similarly to the finite-dimensional case, this Lax matrix satisfies \eqref{cg9} with the $r$-matrix \eqref{rg1}, and hence gives an infinite number of commuting Hamiltonians as in (\ref{cg10}).

Finally, let us remark that the described approach works in the same way for more complicated
models including the the elliptic one (\ref{cg2}).

\subsection{Equations of motion}
As in the finite-dimensional case the Hamiltonians
\beq\label{rg11}
\displaystyle{
		H_{i} = \frac{1}{2}\,\res_{z = z_{i}}\tr L^{2}(z)\, =  \sum\limits_{k:k \neq i}^M\sum\limits_{\vec{a} \in \mZ^{2}}\frac{S_{\vec{a}}^{i}S^{k}_{-\vec{a}}}{z_{i}-z_{k}}\,.
}
\eq
provide equations of motion
\beq\label{rg12}
\displaystyle{
	\partial_{t_{a}}S^{a} = \frac{1}{(4 \pi \h)^{2}}\sum\limits_{k:k \neq a}^M \frac{[S^{a},S^{k}]}{z_{a}-z_{k}}\,,
}
\eq
\beq\label{rg13}
\displaystyle{
	\partial_{t_{a}}S^{k} = -\frac{1}{(4 \pi \h)^{2}}\frac{[S^{a},S^{k}]}{z_{a}-z_{k}}\,,\,\,\,\, k\neq a\,,
}
\eq
which are represented in the Lax form (\ref{q02})
 with the Lax matrix \eqref{rg3} and
	\beq\label{rg14}
	\displaystyle{
		M_{a}(z) = \frac{1}{(4\pi \h)^{2}} \frac{S^{a}}{z-z_a}
	}
	\eq
corresponding to dynamics with respect to $H_a$.

Similarly, in the field theory formulation we deal with the Lax pair
\beq\label{rg15}
\displaystyle{
	L(z,\vec{\phi}) = \sum\limits_{k = 1}^M\frac{S^{k}(\vec{\phi})}{z-z_{k}}\,,
\qquad
  M_{i}(z,\vec{\phi}) = \frac{S^{i}(\vec{\phi})}{z-z_{i}}
}
\eq
and the set of Hamiltonians
\beq\label{rg17}
\displaystyle{
	H_{j} = \frac{1}{2}\,\res_{z = z_{j}}\tr L^{2}(z)\, = \frac{1}{4 \pi^{2       }}\sum\limits_{k:k \neq j}^M\int_{\mathbb{T}^{2}}\mathrm{d}\vec{\phi}\,\frac{S^{j}(\vec{\phi}) \star S^{k}(\vec{\phi})}{z_{j}-z_{k}}\,.
}
\eq
The equations of motion
\beq\label{rg20}
\displaystyle{
	\partial_{t_{a}}S^{a}(\vec{\phi}) = \frac{i}{4\pi \h}\sum\limits_{k \neq a}\frac{[S^{a}(\vec{\phi}),S^{k}(\vec{\phi})]_{\star}}{z_{a}-z_{k}}\,,
}
\eq
\beq\label{rg19}
\displaystyle{
	\partial_{t_{a}}S^{k}(\vec{\phi}) = \frac{i}{4\pi \h}\frac{[S^{k}(\vec{\phi}),S^{a}(\vec{\phi})]_{\star}}{z_{a}-z_{k}}\,,\,\,\,\,\, k\neq a
}
\eq
are generated by the Hamiltonians (\ref{rg17}) and the Poisson brackets as in (\ref{n28})
	\beq\label{rg191}
	\displaystyle{
		\{S^a(\vec{\phi}), S^b(\vec{\theta})\} = \frac{i\delta^{ab}}{4 \pi \h}[S^a(\vec{\phi}),\delta(\vec{\theta}-\vec{\phi})]_{\star}\,.
	}
	\eq
Equations (\ref{rg20})-(\ref{rg19}) are written in the Lax form as follows:
\beq\label{rg18}
\displaystyle{
	\partial_{t_{j}}L(z,\vec{\phi}) = \frac{i}{4 \pi \h}[M_j(z,\phi),L(z,\phi)]_{\star}\,.
}
\eq

\paragraph{The limit $\h \to 0$.}
In the limiting case $\h \to 0$
we come to
	2+1 field theory on $\mathbb{T}^{2} \times \mathbb{R}$ with $M$ fields $S^{k}$ and the Poisson brackets
	\eqref{n35}. The Hamiltonians take the form
	\beq\label{rg21}
	\displaystyle{
		H_{i}^{0} = \frac{1}{4 \pi^{2       }}\sum\limits_{k:k \neq i}^M\int_{\mathbb{T}^{2}}\mathrm{d}\vec{\phi}\,\frac{S^{i}(\vec{\phi}) S^{k}(\vec{\phi})}{z_{i}-z_{k}}
	}
	\eq
	The model is integrable with infinitely many integrals of motion in involution
	\beq\label{rg22}
	\displaystyle{
		Q_{k}^{0}(z) = \frac{1}{4\pi^{2}}\int_{\mathbb{T}^{2}}\mathrm{d}\vec{\theta}\,L^{k}(z,\vec{\theta})\,.
	}
	\eq
	The Lax pair (\ref{rg15}) satisfies the following equation:
\beq\label{rg23}
\displaystyle{
	\partial_{t_{p}}L(z,\vec{\phi}) = \{-M_{p}(z,\vec{\phi}),L(z,\vec{\phi})\}\,,
}
\eq
The equations of motion are as follows:
\beq\label{rg25}
\displaystyle{
	\partial_{t_{a}}S^{a}(\vec{\phi}) = -\sum\limits_{k:k \neq a}^M\frac{\{S^{a}(\vec{\phi}),S^{k}(\vec{\phi})\}}{z_{a}-z_{k}}\,,
}
\eq
\beq\label{rg24}
\displaystyle{
	\partial_{t_{a}}S^{k}(\vec{\phi}) = -\frac{\{S^{k}(\vec{\phi}),S^{a}(\vec{\phi})\}}{z_{a}-z_{k}}\,,\,\,\,\,\, k\neq a
}
\eq
with $\{f,g\} = \partial_{\phi_{1}}f\partial_{\phi_{2}}g-\partial_{\phi_{2}}f\partial_{\phi_{1}}g$.

\section{Large $N$ limit of spectral duality between rational Gaudin models}\label{sec4}
\setcounter{equation}{0}
\subsection{The Adams-Harnad-Hurtubise construction}
Let us briefly recall the main idea underlying the spectral duality between the rational Gaudin models in
the finite-dimensional case \cite{AHH}  (see Theorem 2.1). Consider a pair of ${\rm gl}_N$ and ${\rm gl}_M$ models
with the Lax matrices
\beq\label{sd3}
\displaystyle{L_{ij}(z) = v_{i}\delta_{ij} +
	\sum\limits_{a = 1}^{M}\frac{\xi_{i}^{a}\eta_{j}^{a}}{z-z_{k}}}\,,\quad i,j=1,...,N
\eq
and
\beq\label{sd4}
\displaystyle{
\tilde{L}_{ab}(v) = z_{a}\delta_{ab} +
	\sum\limits_{k = 1}^{N}\frac{\xi_{k}^{a}\eta_{k}^{b}}{v-v_{k}}}\,,\quad a,b=1,...,M
\eq
respectively.
The canonical Poisson structure
\beq\label{sd2}
\displaystyle{
\{\xi^{a}_{i},\eta^{b}_{j}\} = \delta^{ab}\delta_{ij}\,.
}
\eq
provides the Poisson-Lie structure (\ref{cg22}) for the variables $S^a_{ij}=\xi_{i}^{a}\eta_{j}^{a}$ and
similarly for ${\tilde S}^k_{ab}=\xi_{k}^{a}\eta_{k}^{b}$.

The main statement is that the spectral curves of both models (\ref{w03}) and (\ref{w04}) coincide.
The proof is based on the following simple determinant relation.
Consider an $(N+M)\times(N+M)$ matrix $\mats{A}{B}{C}{D}$, where $A$ and $D$ are invertible
 square matrices of sizes $N\times N$ and $M\times M$ respectively,
 while $B$ and $C$ are a rectangular $N\times M$ and $M\times N$ matrices. Then
  \beq\label{sd25}
  \begin{array}{c}
    \displaystyle{
 \det\limits_{(N+M)\times (N+M)}\mats{A}{B}{C}{D}=
 }
 \\ \ \\
     \displaystyle{
 =\det\limits_{N\times N}(A)\det\limits_{M\times M}\Big(D-CA^{-1}B\Big)=
 \det\limits_{M\times M}(D)\det\limits_{N\times N}\Big(A-BD^{-1}C\Big)\,.
 }
 \end{array}
 \eq
By choosing $A=V-v 1_N $ (of size $N\times N$), $B=\xi$ (of size $N\times M$), $C=\eta$ (of size $M\times N$) and $D=Z-z 1_M $ (of size $M\times M$), one gets (all signs are changed in the determinants below)
\beq\label{sd5}
\begin{array}{c}
	\displaystyle{
\det\limits_{N\times N}(v-V)\det\limits_{M\times M}\Big(z-Z-\eta(v-V)^{-1}\xi\Big)
		=\det\limits_{M\times M}(z-Z)\det\limits_{N\times N}\Big(v-V-\xi (z-Z)^{-1} \eta \Big)\,.
	}
\end{array}
\eq
Thus, the spectral curves indeed coincide since\footnote{
From viewpoint of (\ref{sd5}) it is natural to define the dual matrix as $Z+\eta(v-V)^{-1}\xi$, but we use
the transposed one (\ref{sd4}). Obviously, it does not effect the determinant relation.}
\beq\label{sd51}
\begin{array}{c}
	\displaystyle{
L(z)=V+\xi (z-Z)^{-1} \eta\in{\rm Mat}_N\,,\qquad {\ti L}^T(v)=Z+\eta(v-V)^{-1}\xi\in{\rm Mat}_M
	}
\end{array}
\eq
and the last  equality in (\ref{sd5}) means
\beq\label{sd6}
\begin{array}{c}
\displaystyle{
\frac{\det\limits_{N\times N}\Big(v-L(z)\Big)}{\det\limits_{N\times N}\Big(v-V\Big)}
			=
\frac{\det\limits_{M\times M}\Big(z-{\ti L}(v)\Big)}{\det\limits_{M\times M}\Big(z-Z\Big)}\,.
}
 \end{array}
\eq

\subsection{Large $N$ limit of rank one matrices}
The main purpose of this paper is to generalize the relation \eqref{sd6} to the field-theoretical case. The large $N$ limit of the spectral duality was already discussed for multi-matrix models \cite{BEHL}, where the curves of the type \eqref{w03} have a different nature and appear in the analysis of the action of differential operators on orthogonal polynomials. Here we apply the approach based on the field-theoretical representation of $A_{\h}$ described in the previous sections.

 The determinant relation (\ref{sd6}) is based on the property that the residues (\ref{w05})-(\ref{w06}) of
the Lax matrices are of rank 1. Let us formulate this property for the field variables.
For the functions on $\mathbb{T}^{2}$ (elements of $A_{\h}$) there is no canonical way to define a ''rank 1 function''. At the same time, an infinite-dimensional matrix of this type of course exists. Let the trace of
\beq\label{sd12}
\displaystyle{S_{ij} = \xi_{i}\eta_{j} \,,}
\eq
converge. Notice that \eqref{sd12} is in $\mathrm{Mat}_{\infty}(\mathbb{C}[[\h^{-1},\h]])$ only if finitely many $\xi_{i}$ and $\eta_{j}$ are not equal to zero. To calculate matrix powers of \eqref{sd12} with infinite number of nonzero elements one needs to impose some convergence condition, which in this case are equivalent to the existence of a convergent trace. To fulfill this condition we regard vector $\xi$ and covector $\eta$ as integrable functions on a circle $S^{1}$:
\beq\label{sd9}
\displaystyle{\xi(\phi) = \sum\limits_{n \in \mZ}\xi_{n}e^{in\phi}}\,,
\eq
\beq\label{sd10}
\displaystyle{\eta(\phi) = \sum\limits_{m \in \mZ}\eta_{m}e^{im\phi}}\,.
\eq
Then the trace of \eqref{sd12} is finite if
\beq\label{sd14}
\displaystyle{
	\int_{S^{1}}\mathrm{d}\theta\, \xi(\theta)\eta(-\theta) < \infty\,,
}
\eq
which is indeed true.

Introduce also the Poisson structure:
\beq\label{sd11}
\displaystyle{
	\{\xi(\phi),\eta(\theta)\} = \delta(\phi + \theta)\,,
}
\eq
where $\delta(\phi)$ is a delta-function on a circle
\beq\label{sdd5}
\displaystyle{
	\delta(\phi) = \sum\limits_{m \in \mZ}e^{im\phi}\,.
}
\eq
Nevertheless, \eqref{sd12} is not in $\mathrm{Mat}_{\infty}(\mathbb{C}[[\h^{-1},\h]])$ we  use the change of basis \eqref{n10}, and then
 write the corresponding function on $\mathbb{T}^{2}$ with the Fourier modes
\beq\label{sd13}
\displaystyle{
S_{\vec{m}} = \sum\limits_{i,j \in \mZ}W_{\vec{m},ij}\xi_{i}\eta_{j}\,.
}
\eq
\begin{predl}\label{lrsd1}
The field $S(\vec\phi)$ with the Fourier modes (\ref{sd13}) is represented in the form
	\beq\label{sdd1}
	\displaystyle{
		S(\vec{\phi}) = \sum\limits_{\vec{m} \in \mZ^{2}}S_{\vec{m}}e^{i\vec{m}\vec{\phi}}
 =  \frac{i}{4\pi \h}\,\eta(\phi_{2}) \star \delta(\phi_{1}) \star \xi(-\phi_{2})\,,
	}
	\eq
The Poisson brackets (\ref{sd11}) provide the Poisson structure (\ref{n28}):
	\beq\label{rsd9}
	\displaystyle{
		\left\{S(\vec{\phi}),S(\vec{\theta})\right\} = \frac{i}{4\pi \h}\left[S(\vec{\phi}),\delta(\vec{\phi} - \vec{\theta})\right]_{\star}\,.	
	}
	\eq
\end{predl}
\begin{proof}
Plugging \eqref{n10} into \eqref{sd13} one gets
\beq\label{sdd2}
\displaystyle{
	S(\vec{\phi}) = \sum\limits_{\vec{m} \in \mZ^{2}}S_{\vec{m}}e^{i\vec{m}\vec{\phi}} = \frac{i}{4\pi \h}\sum\limits_{\vec{m} \in \mZ^{2}}e^{i\vec{m}\vec{\phi}}\om^{\frac{m_{1}m_{2}}{2}}\sum\limits_{a\in \mZ} \om^{m_{1}a}\xi_{a}\eta_{a+m_{2}}\,.
}
\eq
Next, by shifting $m_{2} \to m_{2} - a$ one gets
\beq\label{sdd3}
\displaystyle{
	S(\vec{\phi}) = \sum\limits_{\vec{m} \in \mZ^{2}}S_{\vec{m}}e^{i\vec{m}\vec{\phi}} = \frac{i}{4\pi \h}\sum\limits_{\vec{m} \in \mZ^{2}}e^{i\vec{m}\vec{\phi}}\om^{\frac{m_{1}m_{2}}{2}}\sum\limits_{a\in \mZ} e^{-ia\phi_{2}}\om^{\frac{m_{1}a}{2}}\xi_{a}\eta_{m_{2}}\,.
}
\eq
Using \eqref{n18}, combine the powers of $\om$ into the star-products:
\beq\label{sdd4}
\displaystyle{
S(\vec{\phi}) = \frac{i}{4\pi \h}\sum\limits_{m_{1},m_{2},a \in \mZ}e^{im_{2}\phi_{2}}\star e^{im_{1}\phi_{1}}\star e^{-ia\phi_{2}} \xi_{a}\eta_{m_{2}}\,.
}
\eq
In this way \eqref{sdd1} follows from the associativity and linearity of $\star$.

In terms of Fourier modes the Poisson brackets (\ref{sd11}) means
	\beq\label{rsd10}
	\displaystyle{
		\{\xi_{m},\eta_{k}\} = \delta_{mk}\,,\quad m,k\in\mZ\,.
	}
	\eq
	Therefore, \eqref{rsd10} defines the Poisson-Lie bracket on $\mathfrak{gl}_{\infty}^{*}$ due to
	\beq\label{rsd11}
	\displaystyle{
		\{\xi_{i}\eta_{j},\xi_{k}\eta_{l}\} = \delta_{il}\,\xi_{k}\eta_{j} - \delta_{kj}\,\xi_{i}\eta_{l}\,.
	}
	\eq
	Since $W_{\vec{m},ij}$ is a homomorphism of associative algebras we conclude
	\beq\label{rsd13}
	\displaystyle{
		\left\{S_{\vec{m}},S_{\vec{n}}\right\} =  \frac{1}{2\pi\h}\sin(2\pi \h \vec{m}\times \vec{n})S_{\vec{m}+\vec{n}}\,,
	}
	\eq
	which coincides with \eqref{rsd9}.
\end{proof}

Traces of powers of \eqref{sdd1} contain divergent integrals with $\delta(\phi_{1})^{k}$. To make sense of this integrals, we calculate them in Fourier modes and formally set
\beq\label{add1}
\displaystyle{
	\frac{i}{4\pi\h} \sum_{k \in \mZ}\om^{k(a-b)} = \delta_{ab}\,.
}
\eq
Let us illustrate this by calculating $\tr S \star S$: 
\beq\label{add2}
\begin{array}{c}
	\displaystyle{
	\frac{4 \pi \h}{i}\tr S \star S = \frac{i}{4 \pi \h}\sum\limits_{\vec{m},\vec{n}\in  \mZ^{2}}\sum\limits_{a,b\in  \mZ}\om^{\frac{m_{1}m_{2}}{2}+\frac{n_{1}n_{2}}{2}}\om^{m_{1}a+n_{1}b}\xi_{a}\eta_{a+m_{2}}\xi_{b}\eta_{b+n_{2}}\frac{1}{4 \pi^{2}} \int_{\mathbb{T}^{2}}\mathrm{d}\vec{\phi}\,e^{i \vec{m}\vec{\phi}}\star e^{i \vec{n}\vec{\phi}} = 
		
	}
	\\
	\displaystyle{
		= \sum\limits_{m_{1},k_{1},k_{2}\in \mZ} \eta_{k_{1}}\xi_{m_{1}+k_{1}}\eta_{k_{2}}\xi_{k_{2}-m_{1}} \frac{i}{4\pi \h}\sum\limits_{m_{2}\in \mZ} \om^{m_{2}(m_{1}+k_{1} - k_{2})} = 
	}
	\\
	\displaystyle{
		= \bigg(\sum\limits_{k\in \mZ} \xi_{k}\eta_{k}\bigg)^{2}  = \bigg(\int_{S^{1}}\mathrm{d}\phi\, \xi(\phi)\eta(-\phi)\bigg)^{2}\,.
	}

\end{array}
\eq
Finally, in order to write down the field version of the Lax matrix (\ref{sd3}) in the large $N$ limit
we introduce $M$ canonically conjugated functions on $S^{1}$
\beq\label{rsd7}
\displaystyle{
	\{\xi^{a}(\phi),\eta^{b}(\theta)\} = \delta^{ab}\delta(\phi + \theta)\,, \quad
a,b=1,...,M\,.
}
\eq
In this way we come to the rational $A_{\h}$ Gaudin model \eqref{rg15}
with the Lax matrix (or, more precisely, the Lax function) 
\beq\label{rsd18}
\displaystyle{
L(z,\vec{\phi}) = V(\phi_{1}) + \frac{i}{4 \pi \h}\sum\limits_{k = 1}^M\frac{\eta^{k}(\phi_{2}) \star \delta(\phi_{1}) \star \xi^{k}(-\phi_{2})}{z-z_{k}}\,,
}
\eq
where $V(\phi_{1})$ is a constant ''twist function''
\beq\label{sd15}
\displaystyle{V(\phi) = \frac{i}{4\pi\h}\sum\limits_{n,k \in \mZ} v_{k}\om^{nk}e^{i n\phi}
}
\eq
obtained by mapping the matrix $V=\sum\limits_k v_k E_{kk}$ to $V(\phi)$ through (\ref{dn1})
and considering the components in the basis $T_{\vec m}$ as Fourier modes.

\subsection{Dual model}
Consider first the dual Lax matrix (\ref{q03}) or (\ref{sd4}) in the finite-dimensional case.
Assuming  $|v| > |v_{i}|$ we may use the series expansion
\beq\label{rsd28}
\displaystyle{
	\frac{1}{v-v_{i}} = \sum\limits_{k = 1}^{\infty} \frac{v_{i}^{k-1}}{v^{k}}\,,
}
\eq
which represents  the dual Lax matrix (\ref{sd4}) in the form
\beq\label{rsd27}
\displaystyle{
	\tilde{L}_{ab}(v) = z_{a}\delta_{ab} +  \sum\limits_{i=1}^N\sum\limits_{k =1}^{\infty}\frac{v_{i}^{k-1}\xi^{a}_{i}\eta^{b}_{i}}{v^{k}}\,,\quad a,b=1,...,M\,,
}
\eq
or, equivalently,
\beq\label{sd17}
\displaystyle{\tilde{L}(v) = Z+\sum\limits_{k = 1}^{\infty}\frac{\tilde{S}^{[k-1]}}{v^{k}}\in{\rm Mat}_M
}
\eq
with
\beq\label{sd161}
\displaystyle{\tilde{S}^{[k]}_{ab} = \sum\limits_{i=1}^N v^{k}_{i}\xi^{a}_{i}\eta^{b}_{i}
\,,\quad a,b=1,...,M\,.
}
\eq
In this way we come to the Gaudin model with irregular singularities \cite{FFRL}.
The spectral duality for models of this type  was studied in \cite{VY} including degenerations
corresponding to Jordan blocks. We do not consider the gluing of marked points.
Our aim is to describe the large $N$ limit.
In the standard matrix basis $E_{ab}$ the generalization to infinite $N$ of (\ref{sd161})
is straightforward:
\beq\label{sd16}
\displaystyle{\tilde{S}^{[k]}_{ab} = \sum\limits_{i \in \mZ}v^{k}_{i}\xi^{a}_{i}\eta^{b}_{i}\,.
}
\eq
The Poisson brackets are given by the Poisson-Lie structure on the loop coalgebra (restricted to the negative loops):
\beq\label{rsd6}
\displaystyle{
	\{\tilde{S}^{[k]}_{1},\tilde{S}^{[m]}_{2}\} = \left[P_{12},\tilde{S}^{[k+m]}_{1}\right]\,,\quad k,m\geq 0\,.
}
\eq
or
\beq\label{rsd63}
\displaystyle{
	\{\tilde{S}^{[k]}_{ab},\tilde{S}^{[m]}_{cd}\} =
 \delta_{ad}\tilde{S}^{[k+m]}_{cb}-\delta_{bc}\tilde{S}^{[k+m]}_{ad}\,,\quad
 a,b,c,d=1,...,M;\ \quad k,m\geq 0\,.
}
\eq
The latter follows from direct substitution of (\ref{sd16}) and usage of
	\beq\label{rsd101}
	\displaystyle{
		\{\xi^a_{m},\eta^b_{k}\} = \delta^{ab}\delta_{mk}\,,\quad a,b=1,...,M;\ m,k\in\mZ\,,
	}
	\eq
which is equivalent to (\ref{rsd7}).

Writing it in terms of the fields (\ref{sd9})-(\ref{sd10}) and the twist function (\ref{sd15}) we get
\beq\label{rsd16}
\displaystyle{
	\tilde{S}^{[k]}_{ab} = \frac{1}{4\pi^{2}}\int_{\mathbb{T}^{2}}\mathrm{d}\vec{\phi}\,(V(\phi_{1}))^{k}\star \frac{i}{4\pi\h}\,\eta^{b}(\phi_{2})\star \delta(\phi_{1}) \star \xi^{a}(-\phi_{2}),\quad a,b=1,...,M\,.
}
\eq
Indeed, the matrix elements of \eqref{sd16} may be
written as traces of $\mathrm{Mat}_{\infty}(\mC[[\h^{-1},\h]])$ matrices
\beq\label{rsdd17}
\displaystyle{
	\tilde{S}^{[k]}_{ab} = \sum\limits_{i \in \mZ}v^{k}_{i}\xi^{a}_{i}\eta^{b}_{i} =
 \tr\, \Big(||\delta_{ij}v_{i}^{k}||\cdot ||\xi^{a}_{i}\eta^{b}_{j}||\Big)\,,\quad a,b=1,...,M\,.
}
\eq
According to the Proposition \ref{lrsd1} the infinite
matrix $||\xi^{a}_{i}\eta^{b}_{j}||$ ($i,j\in\mZ$) corresponds
to the function
$\frac{i}{4\pi \h}\,\eta^{b}(\phi_{2})\star \delta(\phi_{1}) \star \xi^{a}(-\phi_{2})$,
and the diagonal twist matrix $||v_{i}\delta_{ij}||$ corresponds to \eqref{sd15}.
Due to the definition of the trace \eqref{n20} on the
space $C^{\infty}(\mathbb{T}^{2})[[\h^{-1},\h]]$, we come to \eqref{rsd16}.

\subsection{Statement of duality}
Let us first introduce an infinite-dimensional analogue of the spectral curve \eqref{w03}, which
we write as
\beq\label{w031}
\displaystyle{
 \Gamma_N(v,z)=\det\limits_{N\times N}\Big(1_N-\frac{1}{v}\,L(z)\Big)=0\,.
}
\eq
Let $L(z) \in A_{\h}$ be the Lax matrix of some $A_{\h}$ integrable model.  Using
relation $\det(A) = e^{\tr \ln A}$ we define the expression for its spectral curve (power series)
to be
	\beq\label{sd7}
	\displaystyle{
\Gamma_\infty(v,z)=\exp\Big(\frac{4 \pi \h}{i}\tr \ln_{\star}\Big(I-\frac{1}{v}\,L(z,\vec{\phi})\Big)\Big)
    }
	\eq
or
	\beq\label{sd71}
	\displaystyle{
\Gamma_\infty(v,z)=\exp\Big(-\frac{4 \pi \h}{i}\sum\limits_{k = 1}^{\infty}\frac{\tr L^{\star k}(z)}{k v^{k}}\Big)\,,
\qquad L^{\star k}(z)=\underbrace{L(z)\star\dots\star L(z)}_{\rm k\ times}\,,
}
	\eq
where $\ln_{\star}$ is the logarithm power series with multiplication given by the $\star$-product.
 The expression in the r.h.s. of (\ref{sd71}) is a formal power series in $v^{-1}$, whose coefficients depend on the spectral parameter $z$ and provide conserved charges. Obviously, this expression coincides with the expression for the spectral curve (\ref{w031})
in the finite-dimensional case.

The main result of this section is the following statement generalizing relation (\ref{sd5}) for infinite $N$.
\begin{theorem}\label{t1}
The following identity holds true:
\beq\label{rsd17}
\displaystyle{
	\exp\Big(\frac{4 \pi \h}{i}\tr \ln_{\star}\Big(I-\frac{V(\phi_{1})}{v}\Big)\Big)\det\limits_{M\times M}(z-\tilde{L}(v)) = \det\limits_{M\times M}(z-Z)	\exp\Big(\frac{4 \pi \h}{i}\tr \ln_{\star}\Big(I-\frac{1}{v}\,L(z,\vec{\phi})\Big)\Big)\,.
}
\eq
\end{theorem}
\begin{proof}
Dividing both sides of (\ref{rsd17}) by $\det\limits_{M\times M}(z-Z)$ and plugging ${\ti L}(v)$ from (\ref{sd17}),
rewrite (\ref{rsd17}) in the form of relation
\beq\label{rsd19}
\begin{array}{c}
\displaystyle{
		\exp\Big(\frac{4 \pi \h}{i}\tr \ln_{\star}\Big(I-\frac{V(\phi_{1})}{v}\Big)\Big)
\exp\Big(\tr\ln\Big(I-(z-Z)^{-1}\sum\limits_{k = 1}^{\infty}\frac{\tilde{S}^{[k-1]}}{v^{k}}\Big)\Big)
=
}
\\
\displaystyle{
= \exp\Big(\frac{4 \pi \h}{i}\tr \ln_{\star}\Big(I-\frac{1}{v}L(z,\vec{\phi})\Big)\Big)\,,
}
\end{array}
\eq
which we need to prove.
For this purpose observe that
\beq\label{rsd20}
\begin{array}{c}
\displaystyle{
			\tr \Big((z-Z)^{-1}\sum\limits_{k = 1}^{\infty}\frac{\tilde{S}^{[k-1]}}{v^{k}}\Big)^{n} =
}
		\\ \ \\
		\displaystyle{
=\sum\limits_{k_{1},...,k_{n}=1}^\infty\frac{1}{v^{k_{1}+...+k_{n}}}
\sum\limits_{p_{1},...,p_{n}\in\mZ}\sum\limits_{a_{1},...,a_{n}=1}^M
\frac{v_{p_{1}}^{k_{1}-1}\xi^{a_{1}}_{p_{1}}\eta^{a_{2}}_{p_{1}}}{z-z_{a_{1}}}\dots
\frac{v_{p_{n}}^{k_{n}-1}\xi^{a_{n}}_{p_{n}}\eta^{a_{1}}_{p_{n}}}{z-z_{a_{n}}} =
		}
		\\ \ \\
		\displaystyle{
\sum\limits_{k_{1},...,k_{n}=1}^\infty\frac{1}{v^{k_{1}+...+k_{n}}}
\sum\limits_{p_{1},...,p_{n}\in\mZ}\sum\limits_{a_{1},...,a_{n}=1}^M
\frac{v_{p_{n}}^{k_{n}-1}\xi^{a_{n}}_{p_{n}}\eta^{a_{n}}_{p_{n-1}}}{z-z_{a_{n}}}
\dots
\frac{v_{p_{2}}^{k_{2}-1}\xi^{a_{2}}_{p_{2}}\eta^{a_{2}}_{p_{1}}}{z-z_{a_{2}}}
\frac{v_{p_{1}}^{k_{1}-1}\xi^{a_{1}}_{p_{1}}\eta^{a_{1}}_{p_{n}}}{z-z_{a_{1}}}\,. }
	\end{array}
\eq
The summation over $p_{1},...,p_{n}\in\mZ$ corresponds to the trace of the product of infinite matrices.
By rewriting it in terms of functions as in \eqref{sdd1}, we get
\beq\label{rsdd20}
\displaystyle{
			\tr \Big(((z-Z)^{-1}\sum\limits_{k = 1}^{\infty}\frac{\tilde{S}^{[k-1]}}{v^{k}}\Big)^{n}
 = \frac{4 \pi \h}{i}\tr\Big(\sum\limits_{k =1}^{\infty}\frac{V(\phi_{1})^{k-1}}{v^{k}}\star \frac{i}{4 \pi \h}
 \sum\limits_{a = 1}^M
 \frac{\eta^{a}(\phi_{2})\star \delta(\phi_{1}) \star \xi^{a}(-\phi_{2}) }{z-z_{a}}\Big)^{\star n}\,.
            }
\eq
Since the expression for $\exp(\tr \ln (I-A))$ is independent of the size of $A$ (even if $A$ is infinite-dimensional),
 we come to
\beq\label{rsd21}
\begin{array}{c}
	\displaystyle{
	\exp\Big(\frac{4 \pi \h}{i}\tr \ln_{\star}\Big(I-\frac{V(\phi_{1})}{v}\Big)\Big)\exp\Big(\tr\ln\Big(I-(z-Z)^{-1}\sum\limits_{k = 1}^{\infty}\frac{\tilde{S}^{[k-1]}}{v^{k}}\Big)\Big) =
	}
\end{array}
\eq
$$
	\displaystyle{
		= 	\exp\Big(\frac{4 \pi \h}{i}\tr \ln_{\star}\Big(I-\frac{V(\phi_{1})}{v}\Big)+\tr \ln_{\star}\Big(I-\sum\limits_{k =1}^{\infty}\frac{V(\phi_{1})^{k-1}}{v^{k}}\star \sum\limits_{a = 1}^M\frac{\eta^{a}(\phi_{2})\star \delta(\phi_{1}) \star \xi^{a}(-\phi_{2})}{z-z_{a}}\Big)\Big)\,.
	}
$$
Next, we use the identity
\beq\label{rsd22}
\displaystyle{
	\tr\ln_{\star}(I - A)(I - B) = \tr\ln_{\star}(I - A)+ \tr\ln_{\star}(I - B)
}
\eq
to transform the second line of \eqref{rsd21} as
\beq\label{rsd211}
\begin{array}{c}
	\displaystyle{
	\exp(\frac{4 \pi \h}{i}\tr \ln_{\star}(I-\frac{V(\phi_{1})}{v}))\exp(\tr\ln(I-(z-Z)^{-1}\sum\limits_{k = 1}^{\infty}\frac{\tilde{S}^{[k-1]}}{v^{k}}) =
	}
	\\ \ \\
	\displaystyle{
		= 	\exp\Big(\frac{4 \pi \h}{i}\tr \ln_{\star}\Big(\Big[I-\frac{V(\phi_{1})}{v}\Big] \star
\Big[I-\sum\limits_{k =1}^{\infty}\frac{V(\phi_{1})^{k-1}}{v^{k}}\star
\frac{i}{4 \pi \h}\sum\limits_{a = 1}^M\frac{\eta^{a}(\phi_{2})\star \delta(\phi_{1}) \star \xi^{a}(-\phi_{2})}{z-z_{a}}\Big]\Big)\Big)\,.
	}
\end{array}
\eq
Finally, using another obvious relation
\beq\label{rsd23}
\displaystyle{
	(v-V(\phi_{1}))\star \sum\limits_{k =1}^{\infty}\frac{V(\phi_{1})^{k-1}}{v^{k}} = 1
}
\eq
rewrite the expression of argument of $\ln_\star$ in the last line of \eqref{rsd211}:
\beq\label{rsd24}
\begin{array}{c}
	\displaystyle{
\Big[I-\frac{V(\phi_{1})}{v}\Big] \star
\Big[I-\sum\limits_{k =1}^{\infty}\frac{V(\phi_{1})^{k-1}}{v^{k}}\star
\frac{i}{4 \pi \h}\sum\limits_{a = 1}^M\frac{\eta^{a}(\phi_{2})\star
 \delta(\phi_{1}) \star \xi^{a}(-\phi_{2})}{z-z_{a}}\Big] =
	}
	\\ \ \\
	\displaystyle{
=\frac{v - V(\phi_{1})}{v} - \frac{1}{v}\frac{i}{4 \pi \h}\sum\limits_{a = 1}^M
\frac{\eta^{a}(\phi_{2})\star \delta(\phi_{1}) \star \xi^{a}(-\phi_{2})}{z-z_{a}}
=
I-\frac{1}{v}\,L(z,\vec{\phi})\,.
}
\end{array}
\eq
This finishes the proof of \eqref{rsd19}, and therefore the proof of the Theorem.
\end{proof}

Notice that \eqref{rsd17}  describes the relation between the spectral curve of the dual model \eqref{sd17} and the spectral curve (power series)
 of the original model \eqref{rsd18}. Therefore, we treat the statement of Theorem \ref{t1} as
(the field analogue of) the spectral duality between rational Gaudin models $(\ref{sd3}-\ref{sd6})$ in the large $N$ limit.

A similar approach allows one to describe the large $N$ limit of the duality between the classical generalized (that is twisted and inhomogeneous)  XXX Heisenberg spin chain and the trigonometric Gaudin models as well as the duality
between XXZ generalized spin chains. These results will be given in the forthcoming papers.

\subsection*{Acknowledgments}
This work was supported by the Russian Science Foundation under grant no. 25-11-00081,\\
https://rscf.ru/en/project/25-11-00081/ and performed at Steklov Mathematical Institute of Russian Academy of Sciences.

\begin{small}

\end{small}


\begin{thebibliography}{99}
\addcontentsline{toc}{section}{References}

\bibitem{AHH}
M. Adams, J. Harnad, J. Hurtubise,
{\it Dual moment maps into loop algebras},
Lett. Math. Phys. 20 (1990) 299--308.


\bibitem{Arnold}
V. Arnold, {\it Sur la g\'eom\'etrie diff\'erentielle des groupes de Lie de dimension infinie et ses applications \'a l'hydrodynamique des fluides parfaits},
 Annales de l'Institut Fourier, Vol. 16 no. 1 (1966), 319--361.

V.I. Arnold, B.A. Khesin, {\it Topological methods in hydrodynamics},
Applied Mathematical Sciences, Vol. 125,
Springer, New York, 1998.

\bibitem{MM} H. Awata, H.Kanno, T. Matsumoto, A. Mironov, Al. Morozov, An. Morozov, Y. Ohkubo, Y. Zenkevich,
{\it Explicit examples of DIM constraints for network matrix models},
JHEP 07 (2016) 103; \\	arXiv:1604.08366 [hep-th].

H. Awata, H. Kanno, A. Mironov, A. Morozov, K. Suetake, Y. Zenkevich,
{\it (q,t)-KZ equations for quantum toroidal algebra and Nekrasov partition functions on ALE spaces},
JHEP 03 (2018) 192;\\
	arXiv:1712.08016 [hep-th].

\bibitem{Belavin}
A.A. Belavin,
{\em Dynamical symmetry of integrable quantum systems},
Nucl. Phys. B, 180 (1981) 189--200.

A. Belavin, V. Drinfeld,
{\em Solutions of the classical Yang–Baxter equation for simple Lie algebras},
Functional Analysis and Its Applications, 16:3 (1982) 159--180.

\bibitem{BEHL}
M. Bertola, B. Eynard, J. Harnad,
{\it Duality of Spectral Curves Arising in Two-Matrix Models},
Theoret. and Math. Phys., 134:1 (2003), 27--38; 	arXiv:nlin/0112006 [nlin.SI].

M. Bertola, B. Eynard, J. Harnad,
{\it Duality, Biorthogonal Polynomials and Multi-Matrix Models},
Commun. Math. Phys. 229 (2002) 73-120; arXiv:nlin/0108049 [nlin.SI].

M.T. Luu,
{\it Spectral curve duality beyond the two-matrix model},
J. Math. Phys. 60 (2019) 081702;\\ arXiv:1909.02514 [math-ph].

\bibitem{BHS}
M. Bordemann, J. Hoppe, P. Schailer, M. Schlichenmaier,
{\it ${\rm gl}(\infty)$ and geometric quantization},
Commun. Math. Phys. 138 (1991) 209--244.


\bibitem{DN}
M. Douglas,  N. Nekrasov,
{\it Noncommutative Field Theory},
Rev. Mod. Phys. 73 (2001) 977--1029;\\ arXiv:hep-th/0106048.


\bibitem{FT} L.D. Faddeev, L.A. Takhtajan,
{\em Hamiltonian methods in the theory of solitons},
Springer-Verlag, (1987).



\bibitem{FFZ} D.B. Fairlie, P. Fletcher, C.K. Zachos,
{\it Infinite-dimensional algebras and a trigonometric basis for the classical Lie algebras},
Journal of Mathematical Physics, 31:5 (1990) 1088--1094.

\bibitem{FFRL} B. Feigin, E. Frenkel, N. Reshetikhin,
{\it Gaudin Model, Bethe Ansatz and Critical Level},
Commun. Math. Phys. 166, (1994) 27--62; arXiv:hep-th/9402022.

B. Feigin, E. Frenkel, V.T. Laredo,
{\em Gaudin models with irregular singularities},
Advances in Mathematics,
223:3, (2010) 873--948; 	arXiv:math/0612798 [math.QA].



\bibitem{G}
 R. Garnier,
 {\it Sur une classe de syst\'emes diff\'erentiels ab\'eliens d\'eduits de la th\'eorie des \'equations lin\'eaires},
  Rendiconti del Circolo Matematico di Palermo,
  Vol. 43, (1919) 155--191.


M. Gaudin,
{\it La fonction d’onde de Bethe}, Masson, Paris, 1983.


\bibitem{GVZ} A. Gorsky, M. Vasilyev, A. Zotov,
{\it Dualities in quantum integrable many-body systems and integrable probabilities -- I},
JHEP 04 (2022) 159;
arXiv:2109.05562 [math-ph].

\bibitem{HPS} J. Hoppe,
{\it Quantum theory of a massless relativistic surface and a two-dimensional bound state problem},
preprint, PITHA 86/24: Ph.D.
thesis. M.I.T. (1982).

J. Hoppe,
{\em Lectures on integrable systems}, Vol. 10. Springer Science \& Business Media, 2008.

C. Pope, K. Stelle,
{\it  ${\rm SU}(\infty)$, ${\rm SU}_+(\infty)$ and Area Preserving Algebras},
Phys. Lett. B, 226 (1989) 257--263.




\bibitem{JS}
J. Hoppe, P. Schaller,
{\it Infinitely many versions of ${\rm SU}(\infty)$},
Physics Letters B  237 (1990) 407--410.

\bibitem{HOT}
J. Hoppe, M. Olshanetsky, S. Theisen,
{\it Dynamical systems on quantum tori Lie algebras},
Commun. Math. Phys., 155 (1993) 429--448.



%




\bibitem{MMZZ} A. Mironov, A. Morozov, Y. Zenkevich, A. Zotov,
{\it Spectral duality in integrable systems from AGT conjecture},
JETP letters, 97:1 (2013) 45--51;
arXiv:1204.0913 [hep-th].

\bibitem{MMZZR1}
A. Mironov, A. Morozov, B. Runov, Y. Zenkevich, A. Zotov,
 {\it Spectral duality between Heisenberg chain and Gaudin model},
Lett. Math. Phys., 103 (2013)
299--329; arXiv:1206.6349 [hep-th].

\bibitem{MMZZR2}
A. Mironov, A. Morozov, B. Runov, Y. Zenkevich, A. Zotov,
{\it Spectral dualities in XXZ spin chains and five dimensional gauge theories},
JHEP 1312 (2013) 034; arXiv:1307.1502 [hep-th].


\bibitem{MTV} E. Mukhin, V. Tarasov, A. Varchenko,
{\it Bispectral and $({\rm gl}_N,{\rm gl}_M)$ Dualities}, 	arXiv:math/0510364 [math.QA].

E. Mukhin, V. Tarasov, A. Varchenko,
{\it Bispectral and $({\rm gl}_N,{\rm gl}_M)$ Dualities, Discrete Versus Differential},
Advances in Mathematics, 218 (2008) 216--265; arXiv:math/0605172 [math.QA].

E. Mukhin, V. Tarasov, A. Varchenko,
{\it A generalization of the Capelli identity},
Algebra, Arithmetic, and Geometry, Progr. in
Math. Vol. 270 (2009) 383--398;
arXiv:math/0610799 [math.QA]

%


\bibitem{N1} N. Nekrasov, V. Pestun, S. Shatashvili,
{\it Quantum geometry and quiver gauge theories},
Commun. Math. Phys. 357 (2018) 519--567;
	arXiv:1312.6689 [hep-th].

A. Gadde, S. Gukov, P. Putrov,
{\it Walls, Lines, and Spectral Dualities in 3d Gauge Theories},
JHEP 05 (2014) 047;
arXiv:1302.0015 [hep-th].

\bibitem{Olsh}
M. Olshanetsky,
{\it Large $N$ limit of integrable models},
Moscow Mathematical Journal, 3:4 (2003) 1307--1331;
	arXiv:nlin/0307044 [nlin.SI].

G. Aminov, S. Arthamonov, A. Levin, M. Olshanetsky, A. Zotov,
{\it Painlev\'e Field Theory},\\
	arXiv:1306.3265 [math-ph].

A. Levin, M. Olshanetsky, A. Zotov,
 {\it Relativistic Classical Integrable Tops and Quantum R-matrices},
  JHEP 07 (2014) 012;
arXiv:1405.7523 [hep-th].


\bibitem{PZ}
R. Potapov, A. Zotov,
{\it Interrelations between dualities in classical integrable systems and classical-classical version of quantum-classical duality},
Theoret. and Math. Phys., 222:2 (2025), 252--275;\\ arXiv:2410.19035 [math-ph].



\bibitem{R}
M. Rieffel,
{\it Non-commutatuive tori -- a case study of non-commutatuive differentible manidolds},
Contemp. Math. 105 (1990) 191--211.




\bibitem{VY}
B. Vicedo, Ch. Young
{\it (${\mathfrak{gl}}_M$, ${\mathfrak{gl}}_N$)-Dualities in Gaudin Models with Irregular Singularities},
SIGMA 14 (2018), 040; 	arXiv:1710.08672 [math.QA].







\bibitem{Weyl} H. Weyl, {\it The Theory of Groups and Quantum Mechanics}, Dover, New
York, 1931.

J.E. Moyal,
{\it Quantum mechanics as a statistical theory},
 Mathematical Proceedings of the Cambridge Philosophical Society. 45:1 (1949) 99--124.

F. Bayen, M. Flato, C. Fronsdal, A. Lichnerowicz, D. Sternheimer,
{\it Deformation theory and quantization. I. Deformations of symplectic structures},
Annals of Physics, 111:1 (1978) 61--110.




\end{thebibliography}
\end{document}